\documentclass[12pt]{article}
\usepackage{epsfig}
\usepackage{amssymb, amscd, amsmath, amsthm, amsfonts,latexsym }

\newcommand\btr{\mbox{\boldmath${\triangledown}$}}
\newtheorem{theorem}{Theorem}
\newtheorem{lemma}{Lemma}
\newtheorem{corollary}{Corollary}

\newcommand\btheta{\mbox{\boldmath${\theta}$}}
\newcommand\boldeta{\mbox{\boldmath${\eta}$}}
\newcommand{\bu}{{\bf u}}
\numberwithin{equation}{section} \numberwithin{theorem}{section}
\numberwithin{lemma}{section} \numberwithin{example}{section}
\numberwithin{corollary}{section}
\numberwithin{corollary}{section}
\begin{document}

\def\beq{\begin{equation}}
\def\eeq{\end{equation}}
\def\sign{\text{\rm sign}\,}
\allowdisplaybreaks

\title{\bf Estimation in nonstationary random coefficient autoregressive models}

\author{Istv\'an Berkes$^{\rm (1)}$, Lajos Horv\'ath$^{\rm (2)}$
and Shiqing Ling$^{\rm (3)}$}

\date{}
\maketitle

\begin{abstract}
We investigate the estimation of parameters in the random
coefficient autoregressive model $X_k = (\varphi + b_k) X_{k - 1}
+ e_k$, where $(\varphi, \omega^2, \sigma^2)$ is the parameter of
the process, $Eb^2_0 = \omega^2$, $E e^2_0 = \sigma^2$. We
consider a nonstationary RCA process satisfying $E \log |\varphi
+ b_0| \geq 0$ and show that $\sigma^2$ cannot be estimated by the
quasi-maximum likelihood method. The asymptotic normality of the
quasi-maximum likelihood estimator for $(\varphi, \omega^2)$ is
proven so the unit root problem does not exist in the random
coefficient autoregressive model.



\end{abstract}

\renewcommand{\thefootnote}{}
\footnote{$^{(1)}$Department of Statistics, Technical University
Graz, Steyrergasse 17/IV, A--8010 Graz, Austria. Research
partially supported by OTKA grants T 43037 and K 61052.}
\footnote{$^{(2)}$Department of
Mathematics, University of Utah, 155 South 1400 East Salt Lake
City, UT 84112-0090, USA. Research partially supported by NSF
grant DMS 0604670 and grant RGC-HKUST 6428/06H.}
\footnote{$^{(3)}$Department of Mathematics and Statistics,
University of Science and Technology, Clear Water Bay Kowloon,
Hong Kong. Research partially supported by Hong Kong Research Grants Council}

\section{Introduction}
\label{s:1}

In this paper we are interested in the random coefficient model
(RCA) defined by the equations
\begin{equation}
X_k = (\varphi + b_k) X_{k - 1} + e_k,\;\;\; -\infty<k<\infty ,\label{eq:1.1}
\end{equation}
where $\varphi$ is a real parameter. The RCA process was
introduced by And\'el (1976) who also studied its
properties. For a detailed early study we refer to Nicholls and Quinn (1982). Throughout this paper we assume that
\begin{equation}
\{(b_k,e_k)\} \text{ are independent, identically distributed random
vectors.} \label{eq:1.2}
\end{equation}
Let $\log^+x=\max\{\log x, 0\}$. It follows from  Aue et al.\ (2006) (cf.\ also Quinn (1980, 1982)) that
under condition \eqref{eq:1.2} and
\begin{equation}
E\log^+|e_0|<\infty \;\;\;\mbox{and}\;\;\; E\log^+|\varphi+b_0|<\infty, \label{n1}
\end{equation}
equation
\eqref{eq:1.1} has a stationary, nonanticipating (i.e.\ $X_k$ is
measurable with respect to the $\sigma$--algebra generated by $(b_i,
e_i),i \leq k$) if and only if
\begin{equation}
-\infty \leq E \log | \varphi + b_0| < 0. \label{eq:1.5}
\end{equation}

Quinn and Nicholls (1981) started the study of the estimation of
the parameter of the process in \eqref{eq:1.1}. Let $ \btheta
= (\varphi, \omega^2, \sigma^2)$, where
\begin{equation}
Eb_0 = 0, \quad Eb^2_0 = \omega^2>0, \label{eq:1.6}
\end{equation}
\begin{equation}
Ee_0 = 0, \quad Ee^2_0 = \sigma^2 > 0 \label{eq:1.7}
\end{equation}
and
\begin{equation} \mbox{cov}(b_0,e_0)=0.\label{n2}
\end{equation}
Aue et al. (2006) used the quasi--maximum likelihood method to
estimate $\btheta$ when \eqref{eq:1.5} holds. They
established the strong consistency as well as the asymptotic
normality of the quasi--maximum likelihood estimator under minimal
conditions.

In this paper we consider the case when \eqref{eq:1.5} does not
hold. We assume
\begin{equation}\label{def}
 X_k = (\varphi + b_k) X_{k - 1} + e_k, \;\;\;1\leq k\leq n
\end{equation}
and
\begin{equation}\label{cond}
E \log | \varphi + b_0| \geq 0,
\end{equation}
i.e.\ we start the recursion in \eqref{def} from
the initial value $X_0 $ and \eqref{cond} guarantees that the solutions of \eqref{def} cannot converge. Throughout  this paper we assume that $X_0$ is a constant. Following the theory developed for the stationary case, we estimate the parameter $\btheta$ of the process in \eqref{def} using the quasi--likelihood method.  Assuming that $b_0$ and $e_0$ are
normally distributed, the conditional log--likelihood function (the
constant terms are omitted) is given by
\[
L_n({\bf u}) = \sum^n_{k = 1} \ell_k({\bf u}) \ \text{ with }
\ell_k({\bf u}) = -\frac12 \left(\log (x X^2_{k - 1} + y) +
\frac{(X_k - s X_{k - 1})^2}{xX^2_{k - 1} + y}\right),
\]
where ${\bf u} = (s,x, y)$. We show that
$$\frac{1}{n} L_n({\bf u}) \overset{P}{\longrightarrow} \infty$$
but
$$
\frac{1}{n}\bigl(L_n({\bf u}) - L_n(\btheta)\bigr) \overset{P}
{\longrightarrow} f(s,x)
\qquad\text{for all} \;\;\bu\;\;\mbox{with}\;\;x>0\;\;\mbox{and}\;\;y>0,
$$
where
\beq f(s,x) =
\frac12 \left\{ \log \frac{\omega^2}{x} + 1 - \frac{\omega^2}{x} -
\frac{(\varphi - s)^2}{x} \right\}. \label{eq:1.8}
\eeq
Since
$f(\cdot)$ does not depend on $y$, the quasi--maximum likelihood
method cannot be used to estimate $\sigma^2$. Since $|X_n|
\overset{P}{\longrightarrow} \infty \;(n\to \infty)$ (cf.\ Lemma \ref{lem:4.0}), so in
\eqref{eq:1.1} $b_n X_{n - 1}$ dominates $e_n$ which is the reason
why the variance of $e_0$ cannot be estimated by the
quasi--likelihood method. Hence we are interested in estimating  ${\boldeta} = (\varphi,
\omega^2)$. Now $\widehat{{\boldsymbol \eta}}_n =
\widehat{{\boldsymbol \eta}}_n(y)=(\widehat{\eta}_{n,1}(y), \widehat{\eta}_{n,2}(y))$ is defined by
\[
\max_{{\bf z} \in \Gamma} L_n ({\bf z}, y) =
L_n(\widehat{{\boldsymbol \eta}}_n, y),
\]
${\bf z} = (s,x)$ and the set $\Gamma$ satisfies
\begin{equation} \label{n3}
\Gamma = \bigl\{ (s,x) :\ s_* \leq s \leq s^*, \ x_* \leq x \leq
x^* \bigr\}
\end{equation}
with some $s_*<s^*$, $0 < x_*<x^*$. We prove the asymptotic
consistency of $\widehat { \boldeta}_n(y)$ for all~$y$ and consider
the asymptotic normality of $\widehat { \boldeta}_n$ under various
conditions.

\section{Results}
\label{sec:2}

First  we study the asymptotic consistency  of
$\widehat {\boldsymbol \eta}_n(y)$.
\begin{theorem}\label{thm-new}
If \eqref{eq:1.2}, \eqref{eq:1.6}--\eqref{cond} and\eqref{n3} hold,
then
\begin{equation} \label{eq:2.8}
\widehat{\boldsymbol \eta}_n(y) \overset{P}{\longrightarrow}
(\varphi, \omega^2)
\end{equation}
for all $y>0$.
\end{theorem}


Next we consider the asymptotic normality of $\widehat {\boldsymbol \eta}_n(y)$.  Let
\begin{equation}\label{omega}
\Omega_0 = \left(\begin{matrix}
\omega^2 & \omega^2Eb^3_0
\vspace{.3 cm}\\
\omega^2Eb^3_0 & \mbox{var}(b^2_0)
\end{matrix} \right).
\end{equation}

\begin{theorem}
\label{th:2.4}
 If the conditions of Theorem \ref{thm-new} are
 satisfied and
\beq\label{moment}
 E e^4_0 < \infty \quad \text{ and } \quad Eb^4_0 < \infty,
\eeq
then the distribution of $ n^{1/2}(\widehat{\boldsymbol\eta}_n (\sigma^2) - (\varphi,\omega^2))$
converges to the bivariate normal distribution  with mean $\bf 0$ and covariance
matrix~$\Omega_0$.

\end{theorem}

We note that Theorems~ \ref{thm-new} and  \ref{th:2.4} were
obtained by Ling and Li (2006) as a preliminary result for the study of non-stationary double AR(1) processes  when $b_0$ and $e_0$ are normally
distributed and independent. Their result implies that in case of normal $(b_0, e_0)$,  $\sigma^2$ cannot be
estimated by the quasi--maximum likelihood method. A similar
phenomenon was also observed by Jensen and Rahbek (2004a,b) in
nonstationary ARCH models. Theorem \ref{th:2.4} assumes that
$\sigma^2$ is known. %
 We show in the next
section that $\widehat{\boldsymbol\eta}_n (y)$ is asymptotically
normal for all $y>0$ under the  condition $E\log |\varphi+b_0| >0$.\\

Usually, the statistical inference is about $\varphi$, the
expected value of the autoregressive coefficient. We show that
$\widehat{\eta}_{n,1}(y)$ is asymptotically normal for all $y$ so there is
no need to know $\sigma^2$ to get asymptotic statistical inference
about $\varphi$.
\begin{theorem}
\label{th:2.5}
 We assume that the conditions of Theorem \ref{thm-new} are
 satisfied and \eqref{moment} holds. Then for any $y>0$ the distribution of  $\sqrt{n}({\widehat{\eta}_{n,1} (y) -
\varphi})/{\omega}$ converges to the standard normal distribution  
and consequently the distribution of $\sqrt{n}({\widehat{\eta}_{n,1} (y) -
\varphi})/{\sqrt{\widehat{\eta}_{n,2}(y)}}$ converges also to the standard normal distribution.

%
\end{theorem}

 Next we are
interested in the asymptotic distribution of
$\widehat{\boldsymbol\eta}_n (\sigma^2) - (\varphi, \omega^2)$
without assuming \eqref{moment}. The assumption $Eb^4_0<\infty$
will be replaced with the requirement that $b_0^2$ is in the
domain of attraction of a stable law.  This means
that

\beq\label{stab1} P\{b_0^2>x\}= x^{-\alpha}L(x),\;\;\mbox{where  }
1< \alpha<2\;\;
 \mbox{and }
L\;\mbox{is a slowly varying function at } \infty .
\eeq
Asumption $\alpha>1$ guarantees that $Eb_0^2=\omega^2$ exists. Let
$$
a_n=\inf\{x: x^{-\alpha}L(x)\leq 1/n\}.
$$
If \eqref{stab1} holds, then
\beq\label{stab2} \frac{1}{a_n}\sum_{1\leq i \leq
n}(b_i^2-\omega^2)\overset{\mathcal D}{\longrightarrow} \xi , \eeq
where $\xi$ is a stable random variable with characteristic
function
\beq\label{char}
\exp\{-d|t|^\alpha(1+{\bf i}\mbox{sign}(t)\tan
(\pi\alpha/2))\},\;\;\;\;\mbox{if}\;\;1<\alpha<2,
\eeq
 and $d$ is a positive constant (cf.\ Breiman (1968, p.\ 204).
\begin{theorem}
\label{stable}
 We assume that the conditions of Theorem \ref{thm-new} are
 satisfied, \eqref{stab1} and
 \beq\label{stab3}
 E|e_0|^\nu<\infty\;\;\;\mbox{with some}\;\;\nu>2\alpha/(\alpha-1)
 \eeq
 hold.
 Then $n^{1/2}(\widehat{\eta}_{n,1}(\sigma^2)-\varphi)$ and 
$ n(\widehat{\eta}_{n,2}(\sigma^2)-\omega^2)/a_n$ are asymptotically independent, the distribution of $n^{1/2}(\widehat{\eta}_{n,1}(\sigma^2)-\varphi)$ converges to the normal distribution with mean 0 and variance $\omega^2$ and the distribution of $ n(\widehat{\eta}_{n,2}(\sigma^2)-\omega^2)/a_n$ converges to the stable distribution with characteristic function given in \eqref{char}.
 \end{theorem}

We note that if $\{e_k\}$ and $\{b_k\}$ are independent sequences, then \eqref{stab3} can be replaced with 
$Ee_0^4<\infty$.
\section{Growth of $X_n$}
\label{sec:gro}

We will show in Section \ref{sec:4} (cf. Lemma \ref{lem:4.0}) that
under the conditions of Theorem \ref{thm-new}, $X_n
\overset{P}{\longrightarrow} \infty$. Now we  find the order of
the growth of $X_n$. To state our results we need further
notation. Let
$$\xi_i = \log |\varphi +b_i|,\;\;\;\; S(i) = \xi_1 + \dots + \xi_i\;\;\;\;\mbox{and}\;\;\;\;\gamma_i =
\prod\limits_{1 \leq j \leq i} \sign(\varphi + b_j).
$$
In this section  we consider the case when
\beq\label{b4}
E \log | \varphi + b_0| > 0.
\eeq
\begin{theorem}\label{po} If \eqref{eq:1.2}, \eqref{n1}, \eqref{def} and \eqref{b4} hold, then
$$
e^{-S(n)}\gamma_n X_n{\longrightarrow} X_0+Y\;\;\;\mbox{a.s.}
$$
where
$$
Y=\sum_{1 \leq i < \infty} e^{-S(i )} \gamma_{i } e_i.
$$
\end{theorem}

The random normalization $\exp (-S(n))$ is the correct one in Theorem \ref{po}, if the limit is non--zero with probability one. The next  result provide conditions for 
\beq\label{goal}
P\{Y+X_0\neq 0\}=1.
\eeq
\begin{theorem}\label{nondeg} We assume that \eqref{eq:1.2}, \eqref{n1}, \eqref{def} and \eqref{b4} hold.\\
(i) If
\beq\label{smooth}
P\{(\varphi+b_0)X_0+e_0=c\}=0\;\;\;\mbox{for all}\;\;c,
\eeq
then \eqref{goal} holds. \\
(ii) If
\beq\label{b1}
\{b_k\} \;\;\mbox{and}\;\;\{e_k\}\;\;\mbox{are independent sequences}
\eeq
and
\beq\label{c0}
P\{e_0=c\}<1\;\;\mbox{for all}\;\;c,
\eeq
then \eqref{goal} holds. 
\end{theorem}
The first corollary says that under condition \eqref{b4}, $X_n$ grows exponentially fast with probability one.
\begin{corollary} \label{co:1}
If \eqref{eq:1.2},\eqref{n1}, \eqref{def} \eqref{b4} and  \eqref{smooth} or \eqref{b1} and \eqref{c0} hold,
then
$$
e^{-\tau n}|X_n|{\longrightarrow}\infty \;\;\;\mbox{a.s.   for all } \;0<\tau<E\log|\varphi+b_0|
$$
and
$$
e^{-\tau n}|X_n|{\longrightarrow} 0 \;\;\;\mbox{a.s.   for all } \;\tau>E\log|\varphi+b_0|.
$$
\end{corollary}

The second corollary is the asymptotic normality of
$\widehat{\boldsymbol\eta}_n (y)$ without assuming that
$y=\sigma^2$.

\begin{corollary} \label{co:2}
If \eqref{eq:1.2}, \eqref{eq:1.6}--\eqref{def}, \eqref{n3}, \eqref{moment},
\eqref{b4} and \eqref{smooth} or \eqref{b1} and \eqref{c0} hold, then for all $y>0$ the distribution of $ n^{1/2}(\widehat{\boldsymbol\eta}_n (\sigma^2) - (\varphi,\omega^2))$
converges to the bivariate normal distribution  with mean ${\bf 0}$ and covariance
matrix~$\Omega_0$.
\end{corollary}

Similarly, in case of $E\log|\varphi+b_0|>0$, we have the following generalization of Theorem \ref{stable}.

\begin{corollary} \label{co:3}
If \eqref{eq:1.2}, \eqref{eq:1.6}--\eqref{def}, \eqref{n3}, \eqref{stab1}, \eqref{stab3},
\eqref{b4} and \eqref{smooth} or \eqref{b1} and \eqref{c0} hold, then for all $y>0$,
$n^{1/2}(\widehat{\eta}_{n,1}(\sigma^2)-\varphi)$ and 
$ n(\widehat{\eta}_{n,2}(\sigma^2)-\omega^2)/a_n$ are asymptotically independent,  the distribution of $n^{1/2}(\widehat{\eta}_{n,1}(\sigma^2)-\varphi)$ converges to the normal distribution with mean 0 and variance $\omega^2$ and the distribution of $ n(\widehat{\eta}_{n,2}(\sigma^2)-\omega^2)/a_n$ converges to the stable distribution with characteristic function given in \eqref{char}.

\end{corollary}
\section{Proofs of Theorems \ref{thm-new}--\ref{stable}}
\label{sec:4}

The proofs will use the following result:

\begin{lemma}\label{lem:4.0} If \eqref{eq:1.2} and \eqref{eq:1.6}--\eqref{cond}  hold,
 then
\beq |X_n| \overset{P}{\longrightarrow} \infty, \label{eq:4.1}
\eeq
\end{lemma}
\begin{proof}
We note that 
$$
P\{e_0+c(\varphi +b_0)=c\}<1\;\;\;\mbox{for all}\;\;c.
$$
Indeed, if $e_0+c(\varphi +b_0)=c$ with probability one, then multiplying this equation with $e_0$ and taking expected values we get $Ee_0^2+c\varphi Ee_0+ cEb_0e_0=cEe_0$. Since $Ee_0=Ee_0b_0=0$, we get $Ee_0^2=0$, which contradicts $Ee_0^2=\sigma^2>0$ (c.f.\ \eqref{eq:1.7}).
Since \eqref{cond} implies $ P\{\varphi+b_0=0\}=0$, the result follows immediately from Remark 2.8 and Corollary 4.1  of Goldie and Maller (2000).
\end{proof}

We start with the study of the log likelihood function.

\begin{lemma}
\label{lem:4.1} If \eqref{eq:1.2}, \eqref{n1} and \eqref{eq:1.6}--\eqref{cond}  are satisfied,
 then
\beq \sup_{{\bf u} \in \Gamma^*} \left| \frac1n
\bigl(L_n({\bf u}) - L_n({\bf \btheta}) \bigr) - f(s, x) \right|
\overset{P}{\longrightarrow} 0, \label{eq:4.2}
\eeq
where
$f(\cdot)$ is defined in \eqref{eq:1.8} and
\[
\Gamma^* = \bigl\{ {\bf u} = (s,x,y) : \ s_* \leq s \leq s^*, \
x_* \leq x \leq x^*, \ y_* \leq y \leq y^* \bigr\},
\]
with $0 < x_*$ and  $0 < y_*$.
\end{lemma}

\begin{proof}
We write
\begin{align}
L_n ({\bf u}) - L_n ({\btheta}) = \frac12 &\sum_{1 \leq k \leq
n} \log \frac{\omega^2 X^2_{k -
1} + \sigma^2}{x X^2_{k - 1} + y} \label{p0}
\vspace{.3 cm}\\
& +  \frac12 \sum_{1 \leq k \leq n} \frac{(X_{k - 1} b_k
+ e_k)^2}{\omega^2 X^2_{k - 1} + \sigma^2} - \frac12 \sum_{1 \leq k \leq n}\frac{\bigl((\varphi
- s) X_{k - 1} + X_{k - 1} b_k + e_k\bigr)^2}{x X^2_{k - 1} + y}
.\notag
\end{align}
Using the mean  value theorem we conclude
\begin{align*}
\left|\log \frac{\omega^2 X^2_{k - 1} + \sigma^2}{x X^2_{k - 1} +
y} - \log \frac{\omega^2}{x} \right| &\leq c_1
\left(\frac{x^*}{\omega^2} + \frac{x^* X^2_{k - 1} + y^*}{\omega^2
X^2_{k - 1} + \sigma^2} \right) \frac1{x_* X^2_{k
- 1} + y_*} \\
&\leq c_2 \frac1{x_* X^2_{k - 1} + y_*} .
\end{align*}
By \eqref{eq:4.1} we have that
\beq\label{p1}
E \frac1{x_* X^2_n + y_*} \to 0
\eeq
and therefore by the Markov inequality
\beq\label{p1/2}
\frac1n \sum_{1 \leq k \leq n} \sup_{{\bf u} \in \Gamma^*} \left|
\log \frac{\omega^2 X^2_{k - 1} + \sigma^2}{x X^2_{k - 1} + y} -
\log \frac{\omega^2}{x} \right| \overset{P}{\longrightarrow} 0.
\eeq
Also,
\begin{align*}
\sum_{1 \leq k \leq n}& \left\{ \frac{(X_{k - 1} b_k +
e_k)^2}{\omega^2 X^2_{k - 1} + \sigma^2} - 1 \right\}\\
&= \sum_{1 \leq k \leq n} (b^2_k - \omega^2) \frac{X^2_{k -
1}}{\omega^2 X^2_{k - 1} + \sigma^2} + \sum_{1 \leq k \leq n}
\frac{e^2_k}{\omega^2 X^2_{k - 1} +
\sigma^2} \\
&\quad + \sum_{1 \leq k \leq n} b_k e_k \frac{2X_{k - 1}}{\omega^2
X^2_{k - 1} + \sigma^2} - \sum_{1 \leq k \leq n}
\frac{\sigma^2}{\omega^2 X^2_{k - 1} + \sigma^2} .
\end{align*}
Similarly to \eqref{p1} we obtain
\[
\frac1n \sum_{1 \leq k \leq n} \frac{\sigma^2}{\omega^2 X^2_{k -
1} + \sigma^2} \overset{P}{\longrightarrow} 0
\]
Since by \eqref{eq:4.1} and the independence of $e_n$ and $X_{n -
1}$ we have
\[
 E \frac{e^2_n}{\omega^2 X^2_{n - 1} + \sigma^2} \to 0,
\]
 thus we get
\[
\frac1n \sum_{1 \leq k \leq n} \frac{e^2_k}{\omega^2 X^2_{k - 1} +
\sigma^2} \overset{P}{\longrightarrow} 0.
\]
Now we write
\begin{align*}
\sum_{1 \leq k \leq n} (b^2_k - \omega^2) \frac{X^2_{k -
1}}{\omega^2 X^2_{k - 1} +\sigma^2}=
\sum_{1 \leq k \leq n} (b^2_k& - \omega^2)\frac{1}{\omega^2}-
\sum_{1 \leq k \leq n} b^2_k \frac{\sigma^2}{\omega^2(\omega^2 X^2_{k - 1} +\sigma^2)}
\vspace{.3 cm}\\
&+\sum_{1 \leq k \leq n}  \frac{\omega^2\sigma^2}{\omega^2(\omega^2 X^2_{k - 1} +\sigma^2)}.
\end{align*}
The weak law of large numbers yields
\beq\label{p2}
\frac{1}{n}\sum_{1 \leq k \leq n} (b^2_k - \omega^2)\frac{1}{\omega^2}
\overset{P}{\longrightarrow} 0.
\eeq
Using now the independence of $b_k$ and $X_{k-1}$ with \eqref{eq:4.1} we obtain
\[
E\frac{1}{n}\sum_{1 \leq k \leq n} b^2_k \frac{\sigma^2}{\omega^2(\omega^2 X^2_{k - 1} +\sigma^2)}\to 0\;\;\mbox{and}\;\;E\frac{1}{n}\sum_{1 \leq k \leq n}  \frac{\omega^2\sigma^2}{\omega^2(\omega^2 X^2_{k - 1} +\sigma^2)}\to 0
\]
and therefore by the Markov inequality and \eqref{p2} we conclude
\beq
\frac{1}{n}\sum_{1 \leq k \leq n} (b^2_k - \omega^2) \frac{X^2_{k -
1}}{\omega^2 X^2_{k - 1} + \sigma^2}\overset{P}{\longrightarrow} 0.
\eeq
By the independence of $(b_k,e_k)$ and $X_{k-1}$ we get
\[
E\left|\frac1n \sum_{1\leq k \leq n} b_k e_k \frac{X_{k - 1}}{\omega^2
X^2_{k - 1} + \sigma^2}\right|
\leq
\frac1n \sum_{1\leq k \leq n} E|b_k e_k| E\left|\frac{X_{k - 1}}{\omega^2
X^2_{k - 1} + \sigma^2}\right|
 \to 0
\]
on account of \eqref{eq:4.1}, resulting in
\[
\frac1n \sum_{1\leq k \leq n} b_k e_k \frac{X_{k - 1}}{\omega^2
X^2_{k - 1} + \sigma^2}\overset{P}{\longrightarrow} 0.
\]
Hence we proved that
\beq\label{e2}
\frac{1}{n}\sum_{1 \leq k \leq n} \left\{ \frac{(X_{k - 1} b_k +
e_k)^2}{\omega^2 X^2_{k - 1} + \sigma^2} - 1 \right\}\overset{P}{\longrightarrow} 0.
\eeq

Next we write
\begin{align*}
&\frac{\bigl(X_{k - 1}(\varphi - s) + X_{k - 1} b_k +
e_k\bigr)^2}{x X^2_{k - 1} + y} \\
&\quad = (\varphi - s)^2 \frac{X^2_{k - 1}}{x X^2_{k - 1} + y} +
b^2_k \frac{X^2_{k - 1}}{x X^2_{k - 1} + y} + e^2_k \frac1{x
X^2_{k -
1} + y}\\
&\qquad + 2(\varphi - s) b_k \frac{X^2_{k - 1}}{x X^2_{k - 1} + y}
+ 2(\varphi - s) e_k \frac{X_{k - 1}}{x X^2_{k - 1} + y} + 2b_k
e_k \frac{X_{k - 1}}{x X^2_{k - 1} + y}.
\end{align*}
Clearly,
\begin{align*}
\noindent  E \sup_{{\bf u} \in \Gamma^*} \biggl| \sum_{1 \leq k \leq n} e_k
b_k \frac{X_{k - 1}}{x X^2_{k - 1} + y}\biggr| &\leq \sum_{1 \leq
k \leq n} E \biggl|b_k e_k
\frac{X_{k - 1}}{x_* X^2_{k - 1} + y_*}\biggr|\\
& = E|e_0 b_0| \sum_{1\leq k \leq n} E \frac{|X_{k -
1}|}{x_* X^2_{k - 1} + y_*}
\end{align*}
and since by \eqref{eq:4.1}
\[
E \frac{|X_n|}{x_* X^2_n + y_*} \to 0,
\]
 the Markov inequality yields
\[
\sup_{{\bf u} \in \Gamma^*} \biggl| \frac1n \sum_{1 \leq k \leq n}
e_k b_k \frac{X_{k - 1}}{x X^2_{k - 1} + y} \biggr|
\overset{P}{\longrightarrow} 0.
\]
Similar arguments give
\[
\sup_{{\bf u} \in \Gamma^*} \biggl| \frac1n |\varphi - s| \sum_{1
\leq k \leq n} e_k \frac{X_{k - 1}}{x X^2_{k - 1} + y} \biggr|
\overset{P}{\longrightarrow} 0.
\]
Next we observe that
\begin{align*}
\sup_{{\bf u} \in \Gamma^*} \biggl| \sum_{1 \leq k \leq n} b_k
\frac{X^2_{k - 1}}{x X^2_{k - 1} + y} \biggr|
& \leq \sup_{{\bf u} \in \Gamma^*} \frac1x \biggl|\sum_{1
\leq k \leq n} b_k \biggr| + \sup_{{\bf u} \in \Gamma^*} \biggl|
\sum_{1 \leq k \leq n} b_k \frac1x \frac{y}{x X^2_{k - 1} + y}
\biggr|
\vspace{.3 cm}\\
& \leq \frac1{x_*} \biggl| \sum_{1 \leq k \leq n} b_k\biggr|
+ \frac{y^*}{x_*} \sum_{1 \leq k \leq n} |b_k| \frac1{x_* X^2_{k -
1} + y_*} .
\end{align*}
By the law of large numbers we have
\[
\frac1n \sum_{1 \leq k \leq n} b_k \overset{P}{\longrightarrow} 0
\]
and the Markov inequality with \eqref{eq:4.1} gives
\[
\frac1n \sum_{1 \leq k \leq n} |b_k| \frac1{x_* X^2_{k - 1} + y_*}
\overset{P}{\longrightarrow} 0.
\]
Similarly,
\[
\sup_{{\bf u} \in \Gamma^*} \frac1n \biggl| \sum_{1 \leq k \leq n}
e^2_k \frac1{x X^2_{k - 1} + y}\biggr| \leq \frac1n \sum_{1 \leq k
\leq n} e^2_k \frac1{x_* X^2_{k - 1} + y_*}
\overset{P}{\longrightarrow} 0.
\]
Now,
\[
b^2_k \frac{X^2_{k - 1}}{x X^2_{k - 1} + y} - \frac{\omega^2}{x} =
-b^2_k \frac{y}{x(x X^2_{k - 1} + y)} + \frac1x (b^2_k -
\omega^2),
\]
and therefore, arguing as above, we get
\begin{align*}
&\sup_{{\bf u} \in \Gamma^*} \frac1n \biggl| \sum_{1 \leq k \leq
n} \biggl(b^2_k \frac{X^2_{k - 1}}{x X^2_{k - 1} + y} -
\frac{\omega^2}{x} \biggr) \biggr| \\
&\qquad \leq \frac1{x_*} \frac{y^*}{n }\sum_{1 \leq k \leq n} b^2_k
\frac1{x_* X^2_{k - 1} + y_*} + \frac1{x_*} \biggl| \sum_{1 \leq k
\leq n} (b^2_k - \omega^2) \biggr| \overset{P}{\longrightarrow} 0.
\end{align*}

Similarly,
\[
\sup_{{\bf u} \in \Gamma^*} \biggl| \frac1n \sum_{1 \leq k \leq n}
\biggl( \frac{X^2_{k - 1}}{x X^2_{k - 1} + y} - \frac1x \biggr)
\biggr| \overset{P}{\longrightarrow} 0.
\]
Thus we proved 
\beq\label{p5}
\sup_{{\bf u} \in \Gamma^*}\left|\frac{1}{n}\sum_{1\leq k \leq n}\frac{\bigl((\varphi
- s) X_{k - 1} + X_{k - 1} b_k + e_k\bigr)^2}{x X^2_{k - 1} + y}-\left(\frac{(\varphi-s)^2}{x}+\frac{\omega^2}{x}\right)\right|\overset{P}{\longrightarrow} 0. 
\eeq
The result in Lemma \ref{lem:4.1} follows from \eqref{p0}, \eqref{p1/2},\eqref{e2} and \eqref{p5}.
\end{proof}

\begin{lemma}
\label{lem:4.2} If the conditions of Lemma~\ref{lem:4.1} are
satisfied and ${\boldsymbol \eta} = (\varphi, \omega^2) \in
\Gamma$, then
\[
\sup_{y_* \leq y \leq y^*} \bigl| \widehat {\boldsymbol \eta}_n(y)
- {\boldsymbol \eta} \bigr| \overset{P}{\longrightarrow} 0\;\;\;\mbox{for all}\;\;0<y_*<y^*.
\]
\end{lemma}

\begin{proof}
It is easy to see that
\[
f(s,x) \leq f({\boldsymbol \eta})\quad \text{ for all } \ (s, x)
\]
and we have equality if and only if $(s,x)={\boldsymbol\eta} $.
Since
\[
\max_{{\bf z} \in \Gamma} \bigl(L_n({\bf z}, y) - L_n({
\btheta})\bigr) = L_n (\widehat{\boldsymbol \eta}_n, y) - L_n({
\btheta}),
\]
$L_n({\bf u})$, ${\bf u} \in \Gamma^*$ is continuous on
$\Gamma^*$, it converges uniformly to $f(s,x)$, standard arguments
provide the result (cf.\ Pfanzangl (1969)).
\end{proof}

\begin{lemma}
\label{lem:4.3} If the conditions of Theorem \ref{th:2.4} are
satisfied, then for all $0 < y$ we have \beq \biggl| g_{1,n}(y) -
\sum_{1 \leq k \leq n} \frac{b_k}{\omega^2} \biggr| = o_P(n^{1/2})
\label{eq:4.10} \eeq
 and
\beq \biggl| g_{2,n}(\sigma^2) - \sum_{1 \leq k \leq n}
\frac1{2\omega^4} (b_k^2-\omega^2 ) \biggr| = o_P(n^{1/2}),
\label{eq:4.11} \eeq
 where $g_{1,n}(y)$ and
$g_{2,n}(y)$ are the partial derivatives of $L_n({\bf u})$ with
respect to $s$ and $x$ at $(\varphi, \omega^2, y)$.
\end{lemma}

\begin{proof}
Elementary calculations yield
\[
\frac{\partial \ell_k({\bf u})}{\partial s} = \frac{(X_k - sX_{k -
1}) X_{k - 1}}{x X^2_{k - 1} + y}
\]
and
\[
\frac{\partial \ell_k({\bf u})}{\partial x} = -\frac12 \left[
\frac{X^2_{k - 1}}{x X^2_{k - 1} + y} - \frac{(X_k - s X_{k -
1})^2 X^2_{k - 1}}{(x X^2_{k - 1} + y)^2} \right]
\]
and therefore
\begin{align*}
g_{1,n}(y) &= \sum_{1 \leq k \leq n} \left\{ \frac{b_k X^2_{k -
1}}{\omega^2 X^2_{k - 1} + y} + \frac{e_k
X_{k - 1}}{\omega^2 X^2_{k - 1} + y} \right\},\\
g_{2,n}(y) &= \sum_{1 \leq k \leq n} -\frac12 \left\{ \frac{X^2_{k
- 1}}{\omega^2 X^2_{k - 1} + y} - \frac{(X_{k - 1} b_k + e_k)^2
X^2_{k - 1}}{(\omega^2 X^2_{k - 1} + y)^2} \right\}.
\end{align*}
Using the independence of $(e_k, b_k)$ and $X_{k - 1}$ we get
\[
\text{\rm var}\biggl( \frac1{n^{1/2}} \sum_{1 \leq k \leq n}
\frac{e_k X_{k - 1}}{\omega^2 X^2_{k - 1} + y} \biggr) =
\frac{\sigma^2}{n} \sum_{1 \leq k \leq n} E\frac{X^2_{k -
1}}{(\omega^2 X^2_{k - 1} + y)^2} \longrightarrow 0.
\]
Similarly,
\begin{align*}
\text{\rm var}&\Biggl( n^{-1/2} \sum_{1 \leq k \leq n} b_k
\biggl\{ \frac{X^2_{k - 1}}{\omega^2 X^2_{k - 1} + y} -
\frac1{\omega^2} \biggr\} \Biggr)
\vspace{.3 cm}\\
&= \text{\rm var} \biggl( n^{-1/2} \sum_{1 \leq k \leq n} b_k
\frac{y}{\omega^2 (\omega^2 X^2_{k - 1} + y)} \biggr)
\vspace{.3 cm}\\
&= \frac{y^2}{\omega^2} \frac1n \sum_{1 \leq k \leq n}
E\frac1{(\omega^2 X^2_{k - 1} + y)^2} \longrightarrow 0,
\end{align*}
and thus an application of the Markov inequality completes the proof of  \eqref{eq:4.10}.

Write
\begin{align}
&\frac{X^2_{k - 1}}{\omega^2 X^2_{k - 1} + y} - \frac{(X_{k - 1}
b_k + e_k)^2 X^2_{k - 1}}{(\omega^2 X^2_{k - 1} + y)^2}
\label{decomp}
\vspace{.3 cm}\\
&= (\omega^2 - b^2_k) \frac{X^4_{k - 1}}{(\omega^2 X^2_{k - 1} + y
)^2} + \frac{X^2_{k - 1}}{(\omega^2 X^2_{k - 1} + y)^2} (y -2e_k
b_k X_{k - 1} - e^2_k) \notag
\vspace{.3 cm}\\
&=(\omega^2 - b^2_k) \frac{X^4_{k - 1}}{(\omega^2 X^2_{k - 1} + y
)^2} + \frac{X^2_{k - 1}}{(\omega^2 X^2_{k - 1} + y)^2}
(\sigma^2-e^2_k) \notag
\vspace{.3 cm}\\
&\hspace{2 cm}-\frac{X^2_{k - 1}}{(\omega^2 X^2_{k - 1} + y)^2}
2e_k b_k X_{k - 1}  \notag
\vspace{.3 cm}\\
& \hspace{2 cm} +\frac{X^2_{k - 1}}{(\omega^2 X^2_{k - 1} +
y)^2}(y-\sigma^2). \notag
\end{align}
One can easily verify
\beq\label{d1}
 E\Biggl( n^{-1/2} \sum_{1 \leq
k \leq n} (\omega^2 - b^2_k) \biggl( \frac{X^4_{k - 1}}{(\omega^2
X^2_{k - 1} + y)^2} - \frac1{\omega^4} \biggr)\Biggr)^2
\longrightarrow 0,
 \eeq
\beq\label{d2}
 E \Biggl( n^{-1/2} \sum_{1 \leq k \leq
n} \frac{X^2_{k - 1}}{(\omega^2 X^2_{k - 1} + y)^2} (\sigma^2-
e^2_k)\Biggr)^2 \longrightarrow 0,
\eeq
\beq\label{d3}
E \Biggl( n^{-1/2} \sum_{1 \leq k \leq n}
\frac{X^2_{k - 1}}{(\omega^2 X^2_{k - 1} + y)^2}  2 e_k b_k X_{k -
1}\Biggr)^2 \longrightarrow 0, \eeq
and since $y=\sigma^2$ is assumed

\beq\label{d4} \frac{1}{n^{1/2}}|y-\sigma^2|E\sum_{1 \leq k \leq
n}\frac{X^2_{k - 1}}{(\omega^2 X^2_{k - 1} + y)^2}= 0, \eeq
and
therefore \eqref{eq:4.11} is proven.
\end{proof}

\begin{lemma}
\label{lem:4.5} If the conditions of Lemma~\ref{lem:4.1} are
satisfied, then
\[
\sup_{{\bf u} \in \Gamma^*} \bigl| g_{ij,n}({\bf u}) - g_{ij}
({\bf u}) \bigr| \overset{P}{\longrightarrow} 0 \quad 1 \leq i, j
\leq 2,
\]
where
\begin{align*}
g_{11, n}({\bf u}) &= \frac{\partial^2}{\partial s^2} \frac1n
L_n({\bf u}) = -\frac1n \sum_{1 \leq k \leq n} \frac{X^2_{k -
1}}{x X^2_{k - 1} + y};\\
g_{12, n}({\bf u}) &= g_{21,n} ({\bf u}) =
\frac{\partial^2}{\partial s \partial x} \frac1n L_n ({\bf u}) =
\frac{\partial^2}{\partial x \partial s} \frac1n L_n({\bf u}) \\
&= \frac1n \sum_{1 \leq k \leq n} - \frac{(X_k - s X_{k - 1})
X^3_{k - 1}}{(x X^2_{k - 1} + y)^2},\\
g_{22,n}({\bf u}) &= \frac{\partial^2}{\partial x^2} \frac1n L_n
({\bf u}) = \frac1n \sum_{1 \leq k \leq n} \left\{\frac{X^4_{k -
1}}{2(x X^2_{k - 1} + y)^2} - \frac{(X_k - s X_{k - 1})^2 X^4_{k -
1}}{(x X^2_{k - 1} + y)^3} \right\}
\end{align*}
and
$$
g_{11}({\bf u}) = -\frac1x, \qquad
g_{12}({\bf u}) = g_{21}({\bf u}) = - \frac{(\varphi - s)}{x^2}, \qquad
g_{22}({\bf u}) = \frac1{2x^2} - \frac{(\varphi - s)^2 +
\omega^2}{x^3}.
$$
\end{lemma}

\begin{proof}
It can be proven along the lines of the proof of
Lemma~\ref{lem:4.1} and therefore the details are omitted.
\end{proof}

\begin{proof}[Proof of Theorem~\ref{th:2.4}]
Combining the central limit theorem for independent identically
distributed random vectors with Lemma~\ref{lem:4.3}, we get that
\beq\label{eq:4.12}
 n^{-1/2}\bigl(g_{1,n}(\sigma^2), g_{2,n}(\sigma^2)\bigr) \overset{\mathcal
D}{\longrightarrow} N_2 ({\bf 0}, \Omega_*),
\eeq
 where
\[
\Omega_* = \left(\begin{matrix}
\displaystyle \frac1{\omega^2}\;\; & \displaystyle \;\;\frac{Eb^3_0}{2\omega^4}
\notag
\vspace{.3 cm}\\
\displaystyle \frac{E b^3_0}{2\omega^4} & \displaystyle \frac{\text{\rm var}\, b^2_0}{4
\omega^8} \end{matrix} \right).\notag
\]
Let $\|\cdot \|$ denote the maximum norm of vectors. Let ${\mathbf \btr}h({\bf u})=(\partial h({\bf u})/\partial u_1,\partial h({\bf u})/\partial u_2)^T. $  Applying the mean
value theorem to the coordinates of $\btr L_n({\bf u}, \sigma^2) $, there are  random vector ${\boldsymbol \xi}_{n,1}$ and  ${\boldsymbol \xi}_{n,2}$  such
that $\| {\boldsymbol \xi}_{n,j} - {\boldsymbol \eta}\| \leq \|
\widehat{{\boldsymbol \eta}}_n - {\boldsymbol \eta}\|, j=1,2$ and
 \beq
{ 0} = \frac{\partial L_n({\boldsymbol \eta},
\sigma^2)}{\partial { u}_j} + \left({\boldsymbol \bigtriangledown} \frac{\partial L_n({\boldsymbol
\xi}_{n,j}, \sigma^2)}{\partial { u}_j }\right)^T
\bigl(\widehat{{\boldsymbol \eta}}_n - {\boldsymbol \eta}),\;\;j=1,2.
\label{eq:4.13} \eeq
 Lemma \ref{lem:4.5} and Theorem \ref{thm-new}
 give that for all $y>0$ 
\beq \displaystyle
\left(
\begin{matrix}
\displaystyle \frac1n {\boldsymbol \bigtriangledown}\frac{\partial
L_n({\boldsymbol \xi}_{n,1}, y)}{\partial { u}_1
},  \frac1n {\boldsymbol \bigtriangledown}\frac{\partial
L_n({\boldsymbol \xi}_{n,2}, y)}{\partial { u}_2}
\end{matrix}
\right)
\overset{P}{\longrightarrow} \Omega_{**},
\label{eq:4.14}
\eeq
where
\[
\Omega_{**} = \left(\begin{matrix}
\displaystyle -\frac1{\omega^2}  & 0
\vspace{.0 cm}\\
0  &  \displaystyle -\frac1{2\omega^4} \end{matrix} \right).
\]
Putting together \eqref{eq:4.12}--\eqref{eq:4.14} we conclude
\[
n^{1/2} \bigl(\widehat{{\boldsymbol \eta}}_n(\sigma^2) -
{\boldsymbol \eta}\bigr) \overset{\mathcal D}{\longrightarrow} N_2
\bigl({\bf 0}, \Omega^{-1}_{**} \Omega_* \Omega^{-1}_{**} \bigr).
\]
Since $\Omega_0 = \Omega^{-1}_{**} \Omega_* \Omega^{-1}_{**}$, the
proof of Theorem~\ref{th:2.4} is complete.
\end{proof}
The proof of Theorem \ref{th:2.5} uses the following lemma.
\begin{lemma}
\label{cont} If the conditions of Lemma~\ref{lem:4.1} are
satisfied, then for all $y>0$
\[
\sup_{{\bf u} \in \Gamma} \left| g_{12,n}({\bf u}) - (s-\varphi)\frac{1}{n}\sum_{1\leq k \leq n}\frac{X_{k-1}^4}{(xX_{k-1}^2+y)^2}\right|={\cal O}_P(n^{-1/2}).
\]
\end{lemma}
\begin{proof}
Using the expression for $g_{12,n}({\bf u})$ in Lemma \ref{lem:4.5} we get that
\begin{align*}
g_{12,n}({\bf u})=-\frac{1}{n}\sum_{1\leq k \leq n}\frac{(\varphi -s)X_{k-1}^4}{(xX_{k-1}^2+y)^2}
-\frac{1}{n}\sum_{1\leq k \leq n}\frac{b_kX_{k-1}^4}{(xX_{k-1}^2+y)^2}-\frac{1}{n}\sum_{1\leq k \leq n}\frac{e_kX_{k-1}^3}{(xX_{k-1}^2+y)^2}.
\end{align*}
Also, 
$$
\frac{1}{n}\sum_{1\leq k \leq n}\frac{b_kX_{k-1}^4}{(xX_{k-1}^2+y)^2}=\frac{1}{n}\frac{1}{x^2}\sum_{1\leq k \leq n}b_k
-\frac{2y}{xn}\sum_{1\leq k \leq n}\frac{b_kX_{k-1}^2}{(xX_{k-1}^2+y)^2}-\frac{1}{n}\frac{y^2}{x^2}\sum_{1\leq k \leq n}\frac{b_k}{(xX_{k-1}^2+y)^2}.
$$
The central limit theorem yields
$$
\sup_{x_*\leq x \leq x^*}\left|\frac{1}{n}\frac{1}{x^2}\sum_{1\leq k \leq n}b_k\right|={\cal O}_P(n^{-1/2}).
$$
Next we show that
\beq\label{conta}
\sup_{x_*\leq x \leq x^*}\left|A_n(x)\right|=o_P(1),
\eeq
where
$$
A_n(x)=\frac{1}{n^{1/2}}\sum_{1\leq k \leq n}\frac{b_kX_{k-1}^2}{(xX_{k-1}^2+y)^2}.
$$
Since for any $x\in [x_*,x^*]$
$$
EA_n(x)^2=\frac{\omega^2}{n}\sum_{1\leq k \leq n}E\left(\frac{X_{k-1}^2}{(xX_{k-1}^2+y)^2}\right)^2\to 0,
$$
the finite dimensional distributions of $A_n(x)$ converge to 0. Similarly, for all $x,x'\in [x_*,x^*]$ we have by the mean value theorem that
\begin{align*}
E(A_n(x)-A_n(x'))^2&=\frac{\omega^2}{n}\sum_{1\leq k \leq n}E\left(   X_{k-1}^2 \left[ \frac{1}{(xX_{k-1}^2+y)^2}-\frac{1}{(x'X_{k-1}^2+y)^2}
\right]\right)^2
\vspace{.3 cm}\\
&\leq (x-x')^2\frac{\omega^2}{n}\sum_{1\leq k \leq n}E\left(  2 X_{k-1}^4\frac{1}{(x_*X_{k-1}^2+y_*)^3}\right)^2
\vspace{.3 cm}\\
&\leq (x-x')^2,
\end{align*}
for all $n$ large enough. By Billingsley (1968, p.\ 96), the sequence $A_n(x)$ is tight, and therefore $A_n(x)$ converges in ${\cal C}[x_*,x^*]$ to 0. Hence the proof of \eqref{conta} is complete. 

Repeating the arguments leading to \eqref{conta}, we conclude 
$$
\sup_{x_*\leq x \leq x^*}\left|\frac{1}{n}\frac{y^2}{x^2}\sum_{1\leq k \leq n}\frac{b_k}{(xX_{k-1}^2+y)^2}\right|+\sup_{x_*\leq x \leq x^*}\left|\frac{1}{n}\sum_{1\leq k \leq n}\frac{e_kX_{k-1}^3}{(xX_{k-1}^2+y)^2}\right|={\cal O}_P(n^{-1/2}).
$$
The proof of Lemma \ref{cont} is established now.
\end{proof}
\begin{proof}[Proof of Theorem \ref{th:2.5}]
Similarly to (\ref{eq:4.13}) we have
$$
 { 0} = \frac{\partial L_n({\boldsymbol \eta}, y)}{\partial
{ u}_j} + \left({\boldsymbol \bigtriangledown}\frac{\partial L_n({\boldsymbol \xi}_{n,j}, y)}{\partial
{ u}_j }\right)^T \bigl(\widehat{{\boldsymbol
\eta}}_n(y) - {\boldsymbol \eta}),\;\;j=1,2,
$$
%
where ${\boldsymbol \xi}_{n,j}$ satisfies $\| {\boldsymbol \xi}_{n,j} -
{\boldsymbol \eta}\| \leq \| \widehat{{\boldsymbol \eta}}_n(y) -
{\boldsymbol \eta}\|$, $j=1,2$. This  gives
\beq\label{eq:4.13a}
\widehat{\eta}_{n,1}(y)-\varphi=-\left(
c_{11}(n)\frac{1}{n}g_{1,n}(y)+c_{12}(n)\frac{1}{n}g_{2,n}(y)
  \right),
\eeq
 where $c_{ij}(n)$ are defined by
\beq\label{inv} \displaystyle 
\left(
\begin{matrix}
\displaystyle \frac1n {\boldsymbol \bigtriangledown}\frac{\partial L_n({\boldsymbol
\xi}_{n,1}, y)}{\partial { u}_1} ,\frac1n {\boldsymbol \bigtriangledown}\frac{\partial L_n({\boldsymbol
\xi}_{n,2}, y)}{\partial { u}_2}
\end{matrix}\right)^{-1}
 = \left(\begin{matrix} \displaystyle c_{11}(n) &
c_{12}(n)
\vspace{.0 cm}\\
c_{21}(n) &  \displaystyle c_{22}(n) \end{matrix} \right).
 \eeq
Using \eqref{decomp}--\eqref{d3} we get that
\beq\label{kell}
g_{2,n}(y)=o_P(n).
\eeq
Now \eqref{eq:4.14} gives that $c_{11}(n)\rightarrow -\omega^2$
in probability. Applying Lemma \ref{cont} and  \eqref{eq:4.14} we
get that
\beq\label{cs}
 |c_{12}(n)|=|\widehat{\eta}_{1,n}(y)-\varphi|{\cal
O}_P(1) +{\cal O}_P(n^{-1/2}).
\eeq
By 
\eqref{eq:4.13a}--\eqref{cs} we conclude
\begin{align*}
\widehat{\eta}_{n,1}(y)-\varphi &=\left({-\omega^2}
+o_P(1)\right)\frac{1}{n}g_{1,n}(y)+[|\widehat{\eta}_{n,1}(y)-\varphi|{\cal O}_P(1) +{\cal
O}_P(n^{-1/2})]o_P(1)
\vspace{.3 cm}\\
&=\left({-\omega^2} +o_P(1)\right)\frac{1}{n}g_{1,n}(y)+|\widehat{\eta}_{n,1}(y)-\varphi|o_P(1) +o_P(n^{-1/2}),
\end{align*}
which yields
$$
\widehat{\eta}_{n,1}(y)-\varphi=(1+o_P(1))^{-1}\left(
({-\omega^2}+o_P(1))\frac{1}{n}g_{n,1}(y)+o_P(n^{-1/2})\right)
$$
Now the first part of Theorem \ref{th:2.5} follows from \eqref{eq:4.10}. 

The second part is an immedaite consequence of the first part
Theorem
\ref{thm-new}
and  Slutsky's lemma.
\end{proof}
The proof of Theorem \ref{stable} is based on the following
modification of Lemma \ref{lem:4.3}.

\begin{lemma}\label{lem:4.6}
If the conditions of Theorem \ref{stable} are satisfied, then for
all $0 < y$ we have \beq \biggl| g_{1,n}(y) - \sum_{1 \leq k \leq
n} \frac{b_k}{\omega^2} \biggr| = o_P(n^{1/2}) \label{eq:4.20}
\eeq
 and
\beq \biggl| g_{2,n}(\sigma^2) - \sum_{1 \leq k \leq n}
\frac1{2\omega^4} (b_k^2-\omega^2 ) \biggr| = o_P(a_n),
\label{eq:4.21} \eeq
 where $g_{1,n}(y)$ and
$g_{2,n}(y)$ are the partial derivatives of $L_n({\bf u})$ with
respect to $s$ and $x$ at $(\varphi, \omega^2, y)$.
\end{lemma}
\begin{proof} We follow the proof of Lemma \ref{lem:4.3}. Since
the proof of \eqref{eq:4.10} required only that $Eb_0^2<\infty$,
we have \eqref{eq:4.20}. 

To prove \eqref{eq:4.21}, we use
\eqref{decomp}. It is assumed that $Ee_0^4<\infty$ and therefore
\eqref{d2} holds. 
Assumption \eqref{stab1} yields that $E|b_0|^{2\tau} <\infty$ for all $0<\tau<\alpha$, and therefore condition \eqref{stab3} with H\"older's inequality gives $Ee_0^2b_0^2<\infty$. Hence
\begin{align*}
E \Biggl( n^{-1/2} \sum_{1 \leq k \leq n}
\frac{X^2_{k - 1}}{(\omega^2 X^2_{k - 1} + y)^2}  2 e_k b_k X_{k -
1}\Biggr)^2 =\frac{4}{n}E(e_0b_0)^2\sum_{1\leq k \leq n}\left(\frac{X^3_{k - 1}}{(\omega^2 X^2_{k - 1} + y)^2}\right)^2
\longrightarrow 0.
\end{align*}
%
 Clearly, \eqref{d4} is satisfied.
Thus it is enough to show that
\beq\label{d5}
 \sum_{1 \leq
k \leq n} (\omega^2 - b^2_k) \biggl( \frac{X^4_{k - 1}}{(\omega^2
X^2_{k - 1} + y)^2} - \frac1{\omega^4} \biggr)=o_P(a_n).
 \eeq

Let
$$
\epsilon_k=b_k^2-\omega^2\;\;\mbox{and}\;\;z_{k-1}=\frac{X^4_{k -
1}}{(\omega^2 X^2_{k - 1} + y)^2} - \frac1{\omega^4}.
$$
It is clear that $|z_k|\leq c_1$ with some constant $c_1$. Also, 
according to Lemma \ref{lem:4.0}, $|z_k|\rightarrow 0$ in
probability, as $k\rightarrow \infty$, and therefore \beq\label{z}
\delta_k=Ez^2_k\rightarrow 0\;\; (k\rightarrow \infty). \eeq
 Fix $n$ and define
 $$
 \epsilon_k^*=\epsilon_kI\{|\epsilon_k|\leq \tau_n
 a_n\}\;\;\;\mbox{and}\;\;\;
 \epsilon^{**}_k=\epsilon^*_k-E\epsilon_k^*,\;\;\;\;1\leq k \leq n,
 $$
 where $\tau_n$ is a numerical sequence (to be chosen later) tending to $\infty$ and
 $I\{\cdot \}$ denotes the indicator function. Let
 $$
 A(t)=\int^t_{-t}x^2dF(x),
 $$
where $F$ denotes the distribution function of $\epsilon_0$. By
the classical theory of the domain of attraction of stable laws
(cf.\ Feller (1966, pp.\ 574--577)) we have that \beq\label{sl1}
\lim_{t\rightarrow \infty}\frac{A(t)}{t^{2-\alpha}L(t)}=c_2 \eeq
with some $0<c_2<\infty$. Also, we note that by the definition of
$a_n$ and the properties of regularly varying functions we get
that \beq\label{sl2} n L(a_n)/a_n^\alpha \rightarrow 1
\;\;\;(n\rightarrow \infty).
 \eeq
We also need that for any $\kappa>0$ there is a constant
$0<c_3<\infty$ such that
\beq\label{sl3} \frac{L(\lambda x)}{L(x)}\leq
c_3\lambda^\kappa\;\;\;\mbox{for all}\;\;\;\lambda\geq
1\;\;\mbox{and}\;\;x\geq 1.
\eeq
The assertion in \eqref{sl3} is an immediate consequence of the
monotone equivalence theorems in  Bingham et al (1987, p.\ 23).
Indeed, there is a non--increasing regularly varying function $\psi$ such that
$$
\lim_{x\rightarrow \infty}\frac{x^{-\kappa} L(x)}{\psi(x)}=1,
$$
and since $\psi(\lambda x)\leq \psi(x)$ for all $\lambda \geq 1$
and $x\geq 1$, so \eqref{sl3} is proven.

Using the independence of $\epsilon_k^{**}$ and $z_{k-1}$ we
conclude
\begin{align*}
E(\epsilon^{**}_kz_{k-1})^2=E(\epsilon^{**}_k)^2Ez_{k-1}^2\leq
E(\epsilon^{*}_k)^2Ez_{k-1}^2=A(\tau_na_n)\delta_{k-1}
\end{align*}
and the orthogonality of $\{\epsilon^{**}_kz_{k-1}, k\geq 1 \}$
yields
$$
\mbox{var}\biggl(\frac{1}{a_n}\sum_{1\leq k \leq
n}\epsilon^{**}_kz_{k-1}\biggl)=\frac{1}{a_n^2}\sum_{1\leq k \leq
n}E(\epsilon^{**}_kz_{k-1})^2\leq
\frac{1}{a_n^2}A(\tau_na_n)\sum_{1\leq k \leq n}\delta_{k-1}.
$$
Combining \eqref{sl1}--\eqref{sl3} we get that
\begin{align*}
\mbox{var}\biggl(\frac{1}{a_n}\sum_{1\leq k \leq
n}\epsilon^{**}_kz_{k-1}\biggl) &={\cal
O}(1)a^{-2}_n(\tau_na_n)^{2-\alpha}L(\tau_n a_n)\sum_{1\leq k \leq
n}\delta_{k-1} \notag
\vspace{.3 cm}\\
&={\cal O}(1)\tau_n^{2-\alpha}\frac{L(\tau_n
a_n)}{L(a_n)}\frac{1}{n}\sum_{1\leq k \leq n}\delta_{k-1}\notag
\vspace{.3 cm}\\
&= {\cal O}(1)\tau_n^{2-\alpha+\kappa}\frac{1}{n}\sum_{1\leq k
\leq n}\delta_{k-1}.\notag
\end{align*}
By \eqref{z}, if $\tau_n\rightarrow \infty $ slowly enough, then
$\tau_n^{2-\alpha+\kappa}\sum_{1\leq k \leq
n}\delta_{k-1}/n\rightarrow 0,$ showing that
$$
\mbox{var}\biggl(\frac{1}{a_n}\sum_{1\leq k \leq
n}\epsilon^{**}_kz_{k-1}\biggl) =o(1).
$$
Using the definitions of $\epsilon^*_k$ and $a_n$ together with \eqref{stab1}, \eqref{sl2} and \eqref{sl3}, we obtain that
$$
\sum_{1\leq k \leq n}P\{\epsilon^*_k\neq \epsilon_k\}=n(1-F(\tau_n
a_n))\rightarrow 0\;\;\;(n\rightarrow \infty).
$$
Next we observe that 
$$
\frac{n}{a_n}\int_{-\tau_n a_n}^{\tau_n a_n}xdF(x)\rightarrow 0
\;\;\;(n\rightarrow \infty),
$$
if $\tau_n\rightarrow \infty $ slowly enough, so using $z_k\to 0\;(k\to \infty)$
we conclude that
$$
\sum_{1\leq k \leq n}z_{k-1}E\epsilon_k^*=\sum_{1\leq k \leq
n}z_{k-1}\int_{-\tau_n a_n}^{\tau_n
a_n}xdF(x)=o(1)n\int_{-\tau_n a_n}^{\tau_n a_n}xdF(x)=o(a_n).
$$
Now the proof of \eqref{d5} is complete.
\end{proof}

\begin{proof}[Proof of Theorem \ref{stable}]
Using \eqref{eq:4.13} we get

\beq\label{eq:4.13b} \widehat{\eta}_{n,1}(\sigma^2)-\varphi=-\left(
c_{11}(n)\frac{1}{n}g_{1,n}(\sigma^2)+c_{12}(n)\frac{1}{n}g_{2,n}(\sigma^2)
  \right),
\eeq

and

\beq\label{eq:4.13c} \widehat{\eta}_{n,2}(\sigma^2)-\omega^2=-\left(
c_{21}(n)\frac{1}{n}g_{1,n}(\sigma^2)+c_{22}(n)\frac{1}{n}g_{2,n}(\sigma^2)
  \right),
\eeq

where $c_{ij}$ are defined in \eqref{inv}. By Lemma \ref{lem:4.6}, 
\eqref{eq:4.14} and  \eqref{inv} we get that
\beq\label{first}
 \frac{n}{a_n}(\widehat{\eta}_{n,2}(\sigma^2)-\omega^2)=\frac{1}{a_n}\sum_{1\leq
 k \leq n}(b_k^2-\omega^2)+o_P(1).
\eeq
Since \eqref{kell} clearly holds, we also have \eqref{cs} and from \eqref{eq:4.21} we obtain that
$$
\widehat{\eta}_{n,1}(\sigma^2)-\varphi=
c_{11}(n)\frac{1}{n}g_{1,n}(\sigma^2)+(|\widehat{\eta}_{n,1}-\varphi|{\cal O}_P(1)+{\cal
O}_P(n^{-1/2})){\cal O}_P(a_n/n).
$$
Hence by \eqref{eq:4.14} and \eqref{eq:4.20} we have
\beq\label{second} n^{1/2}(\widehat{\eta}_{n,1}(\sigma^2)-\varphi)=
n^{-1/2}\sum_{1\leq k \leq n}b_k +o_P(1). \eeq
The convergence in distribution of
$n^{1/2}(\widehat{\eta}_{n,1}(\sigma^2)-\varphi)$ and $n
(\widehat{\eta}_{n,2}(\sigma^2)-\omega^2)/a_n$ now follows from
\eqref{first} and \eqref{second}; only the asymptotic independence
must be established. Note that the vector $(\sum_{1\leq k
\leq}b_k/n^{1/2},\sum_{1\leq k \leq n}(b_k^2-\omega^2)/a_n)$
converges in distribution (cf.\ Section 10.1 in Meerschaert and
Scheffler (2001)). The first coordinate of the limit is normal,
the second does not contain normal component and therefore the
coordinates of the limit distribution are independent (Meerschaert
and Scheffler (2001, p.\ 41)).
\end{proof}

\section{Proofs of Theorem \ref{po} and Corollaries \ref{co:1}--\ref{co:3} }\label{more}

Using \eqref{def} one can easily verify that
\[
X_\ell= \sum^\ell_{i = 1} e_i \prod^\ell_{j = i + 1} (\varphi + b_j)+
X_0\prod^\ell_{j= 1}(\varphi + b_j),
\]
and therefore
\begin{align}\label{basic}
\displaystyle \left({\displaystyle\prod^\ell_{j= 1}(\varphi + b_j)}\right)^{-1}
X_\ell&=\sum^\ell_{i = 1} e_i \left(\prod^i_{j = 1} (\varphi + b_j)\right)^{-1}+X_0
\vspace{.3 cm}\\
&=\sum^\ell_{i = 1} e_i e^{-S(i)}\gamma_i +X_0. \notag
\end{align}

\begin{proof}[Proof of Theorem~\ref{po}]
First we note that assumption \eqref{n1} yields 
$$
|e_i|={\cal O}(e^{ic_1}) \;\;\mbox{a.s. for any }\; c_1>0
$$
(cf.\ Berkes at al (2003)) and therefore by the strong law of large numbers
$$
e^{-S(i)}=o(e^{-ic_2})\;\;\mbox{a.s. for any } 0<c_2<E|\xi_0|.
$$
Hence $Y$ is absolutely convergent with probability one and the result follows immediately from \eqref{basic}.
\end{proof}

The proof of the second part of Theorem \ref{nondeg} is based on the following lemma:
\begin{lemma}\label{levy} If \eqref{eq:1.2},\eqref{n1}, \eqref{def} \eqref{b4}, \eqref{b1} and \eqref{c0} hold, then
\[
P\{ Y = c\} = 0 \quad \text{ for any } \ c.
\]
\end{lemma}
\begin{proof}
First we show that for any sequence $a_n$ \beq \sum_{1 \leq i <
\infty} P \bigl\{ e^{-S(i)} \gamma_{i } e_i \neq a_i \bigm|
\xi_j, -\infty < j < \infty\bigr\} = \infty \;\;\text{ a.s.}
\label{eq:3.2} \eeq Since $E \xi_0$ exists, we get $P \{ \xi_0 =
0\}=0$, so $\gamma_i$ can be $0$ only with probability 0. Hence
\eqref{eq:3.2} holds, if for any sequence $b_n$ \beq \sum_{1 \leq
i < \infty} P \{e_i \neq b_i \} = \infty. \label{eq:3.3} \eeq By
\eqref{b1}, we have \eqref{eq:3.3} if and only if \beq \sum_{1
\leq i < \infty} P \{e_0 \neq b_i \} = \sum_{1 \leq i < \infty}
\bigl( 1 - P\{ e_0 = b_i \} \bigr) = \infty. \label{eq:3.4} \eeq
If $P\{e_0 = b_i\} \to 1$, then $e_0$ must be a constant with
probability 1, contradicting \eqref{c0}.

Using \eqref{eq:3.2} we get that for any sequence $a_n$ \beq
\sum_{1 \leq i < \infty} P \bigl\{ e^{-S(i)} \gamma_{i }
e_i \neq a_i \bigr\} = \infty, \eeq and therefore
Lemma~\ref{levy} follows from L\'evy (1931) (cf.\ also Breiman
(1968, p.~51)).
\end{proof}

\begin{lemma}\label{approx} If \eqref{eq:1.2}, \eqref{eq:1.6}--\eqref{def}, \eqref{n3}, 
\eqref{b4} and \eqref{smooth} or \eqref{b1} and \eqref{c0} hold, then 
\beq \label{eq:5.10}
\biggl| g_{1,n}(y) -
\sum_{1 \leq k \leq n} \frac{b_k}{\omega^2} \biggr| = {\cal O}(1)\;\;
\mbox{a.s.}
\eeq
 and
\beq \label{eq:5.11}
\biggl| g_{2,n}(y) - \sum_{1 \leq k
\leq n} \frac1{2\omega^4} (b_k^2-\omega^2 ) \biggr| =
{\cal O}(1)\;\;
\mbox{a.s.}
\eeq
 where $g_{1,n}(y)$ and
$g_{2,n}(y)$ are the partial derivatives of $L_n({\bf u})$ with
respect to $s$ and $x$ at $(\varphi, \omega^2, y)$.
\end{lemma}
\begin{proof}[Proof of Lemma \ref{approx}]
We return to the decompositions of  $g_{1,n}(y)$ and $g_{2,n}(y)$
used in the proof of Lemma \ref{lem:4.3}. Using Theorem \ref{po}
and Lemma \ref{levy} we get that
\begin{align}
\biggl|\sum_{1\leq k \leq n}& \frac{e_k
X_{k-1}}{\omega^2X^2_{k-1}+y}\biggl|
 \notag
 \vspace{.3 cm}\\
  &\leq
\sum_{1\leq k \leq n}|e_k| \frac{|X_{k-1}|}{\omega^2X^2_{k-1}+y}
\notag
\vspace{.3 cm}\\
 &=\sum_{1\leq k \leq
 n}|e_k|e^{-S(k-1)}\bigl(e^{-S(k-1)}|X_{k-1}|\bigl)\bigl(\omega^2(e^{-S(k-1)}X_{k-1})^2+e^{-2S(k-1)}y\bigl)^{-1}
\notag
\vspace{.3 cm}\\
&\leq \biggl\{\max_{1\leq k
<\infty}\bigl(e^{-S(k-1)}|X_{k-1}|\bigl)\bigl(\omega^2(e^{-S(k-1)}X_{k-1})^2+e^{-2S(k-1)}y\bigl)^{-1}\biggl\}\sum_{1\leq
k \leq
 n}|e_k|e^{-S(k-1)}
\notag
\vspace{.3 cm}\\
&={\cal O}(1)\;\;\;\;\;\mbox{a.s.}\notag
\end{align}
since  by Berkes et al (2003), $\sum_{1\leq k \leq
n}|e_k|e^{-S(k-1)}$ is finite with probability one. Similar
arguments give
$$
\left|\sum_{1\leq k \leq n}b_k\left\{
\frac{X^2_{k-1}}{\omega^2X^2_{k-1}+y}-\frac{1}{\omega^2}\right\}\right|
\leq \sum_{1\leq k \leq
n}|b_k|\frac{1}{\omega^2}\frac{|y|}{\omega^2X^2_{k-1}+y}={\cal
O}(1)\;\;\;\;\;\mbox{a.s.},
$$
completing the proof of \eqref{eq:5.10}.

The proof of \eqref{eq:5.11} goes along the same lines and hence
it is omitted.
\end{proof}

\begin{proof}[Proof of Corollary \ref{co:1}] It is an immediate consequence of the strong law of large numbers and Theorems \ref{po} and \ref{nondeg}.

\end{proof}

\begin{proof}[Proof of Corollary \ref{co:2}] The proof of Theorem
\ref{th:2.4} can be repeated; only Lemma \ref{lem:4.3} must be
replaced with Lemma \ref{approx}.
\end{proof}
\begin{proof}[Proof of Corollary \ref{co:3}] Minor modifications of the proof of  Theorem
\ref{stable} are required only. Namely,   one must use Lemma \ref{approx} instead of Lemma \ref{lem:4.6}.
\end{proof}

\end{document}